\documentclass[lettersize,draftcls]{IEEEtran}
\def\BibTeX{{\rm B\kern-.05em{\sc i\kern-.025em b}\kern-.08em
    T\kern-.1667em\lower.7ex\hbox{E}\kern-.125emX}}

\usepackage{amsmath,amsfonts, amssymb,bbm,amsthm,mathtools,nccmath,bm}

\DeclareMathOperator*{\argmin}{\arg\!\min}
\usepackage{algorithmic}
\usepackage{algorithm}
\usepackage{array}
\usepackage[caption=false,font=normalsize,labelfont=sf,textfont=sf]{subfig}
\usepackage{textcomp}
\usepackage{stfloats}
\usepackage{url}
\usepackage{verbatim}
\usepackage{graphicx}
\usepackage{enumitem}
\usepackage{cite}
\usepackage[normalem]{ulem}
\usepackage{cancel}
\usepackage{booktabs}
\usepackage[table]{xcolor}
\usepackage{hyperref}
\usepackage{tikz}
\usetikzlibrary{automata, positioning}

\newtheorem{theorem}{Theorem}\newtheorem{corollary}{Corollary}
\newtheorem{proposition}{Proposition}\newtheorem{definition}{Definition}\newtheorem{assumption}{Assumption}

\definecolor{color0}{rgb}{0,0.4470,0.7410}   
\definecolor{color1}{rgb}{0.8500,0.3250,0.0980} 
\DeclareRobustCommand{\ActW}{\tikz[baseline=-0.5ex]\fill[color0] (0,0) circle (2pt);}
\DeclareRobustCommand{\ActT}{\tikz[baseline=-0.6ex]\fill[color1] (-2pt,-2pt) rectangle (2pt,2pt);}

\begin{document}

\title{
Minimizing the Age of Incorrect Information of Discrete-Time Markov Sources With Hybrid ARQ
}

\author{Konstantinos Bountrogiannis,~Anthony Ephremides,~Panagiotis Tsakalides,
~and~George Tzagkarakis
\IEEEcompsocitemizethanks{\IEEEcompsocthanksitem Konstantinos Bountrogiannis and Panagiotis Tsakalides are with the Department of Computer Science, University of Crete, Heraklion 700~13, Greece, and with the Institute of Computer Science, Foundation for Research and Technology -- Hellas, Heraklion 700~13, Greece\protect\\
E-mail: kbountrogiannis@csd.uoc.gr;  tsakalid@csd.uoc.gr
\IEEEcompsocthanksitem Anthony Ephremides is with the Electrical and Computer Engineering Department, University of Maryland, College Park, MD 20742 USA\protect\\
E-mail: etony@umd.edu
\IEEEcompsocthanksitem George Tzagkarakis is with the Institute of Computer Science, Foundation for Research and Technology -- Hellas, Heraklion 700~13, Greece\protect\\
E-mail: gtzag@ics.forth.gr}
\thanks{This work was supported in part by the EU under Horizon Europe’s HORIZON-WIDERA-2023-ACCESS-02-01 - Twinning Bottom-Up TwinODIS project (No. 101160009), and the project “Interconnected Smart Cities for Greece 2.0” (No. TAEDR-0536642) under the Action “Flagship actions in interdisciplinary scientific fields of special interest for connection with the productive sector,” as part of Greece 2.0 – National Recovery and Resilience Plan.}}

\markboth{Age of Incorrect Information for Discrete-Time Markov Sources}{}

\maketitle

\begin{abstract}
This work introduces a framework for analyzing the Age of Incorrect Information (AoII) in a real-time monitoring system with a discrete-time Markov source. The sources studied herein are arbitrary finite and irreducible Markov sources. We study a noisy communication system employing a hybrid automatic repeat request (HARQ) protocol, subject to a transmission rate constraint. The optimization problem is formulated as a constrained Markov decision process (CMDP), and it is shown that there exists an optimal policy that is a randomized mixture of two stationary policies. To overcome the intractability of computing the optimal stationary policies, we develop a multiple-threshold policy class where thresholds depend on the source, the receiver, and the packet count. By establishing a Markov renewal-reward structure induced by threshold policies, we derive analytic expressions for the long-term average AoII and transmission rate. The proposed policy is constructed via a relative value iteration algorithm that leverages the threshold structure to skip computations, combined with a bisection search to satisfy the rate constraint. To accommodate scenarios requiring lower computational complexity, we adapt the same technique to produce a simpler single-threshold policy that trades optimality for efficiency. Numerical experiments exhibit that both threshold-based policies outperform periodic scheduling, with the multiple-threshold approach matching the performance of the globally optimal policy.
\end{abstract}

\section{Introduction}\label{sec:introduction}

\IEEEPARstart{R}{eal-time} monitoring of remote data sources plays a central role in the development of many emerging applications, driven by recent advancements in communication and sensing technologies. Examples include autonomous driving, real-time video feedback, anomaly detection in critical infrastructures, remote surgery, augmented reality networks, and haptic communications. In such use cases, the timely delivery of information is essential, and the Age of Information (AoI)~\cite{bib:aoi} metric has been introduced and evaluated to assess system performance. Compared to the traditional delay metrics, AoI measures not only the communication time in the network, but also the time a message spends before being scheduled for submission to the network. As such, AoI measures the freshness of the communicated information and highlights the need to update the monitoring system with more current data to avoid excessive aging. Formally, the AoI process $\delta^{\text{AoI}}_t$ is defined as the difference 

\begin{equation}
    \delta^{\text{AoI}}_t \triangleq t-u_t\,,
\end{equation}
where $u_t$ is the generation time of the most recently received message.

One limitation of the traditional AoI metric is that it focuses solely on the freshness of information, without accounting for the dynamics of the data source itself. For example, if one source changes rapidly and another changes slowly, the same communication system may produce samples with the same AoI for both sources. However, the sample from the fast-changing source may be less relevant or useful for decision-making, because it is more likely to deviate from the contemporary source state at the time of arrival.

This observation prompted the introduction of the Age of Incorrect Information (AoII) metric~\cite{bib:aoii}, which is the focus of this work. Unlike traditional AoI, AoII is a content-aware metric that penalizes the duration during which the information at the monitor differs from the source state, i.e., is incorrect, and is weighted by a distortion function. Specifically, let $g(S_t,W_t)$ denote a distortion function between the source state $S_t$ and the estimate at the monitor $W_t$ at time $t$. Moreover, define the age function
\begin{equation}
    a_t \triangleq t-h_t\,,
\end{equation}
where $h_t$ is the last time instant when $g(S_t,W_t)$ was zero. The AoII process $\delta_t$ is the product
\begin{equation}\label{eq:aoii_def}
    \delta_t \triangleq a_t \cdot g(S_t,W_t).
\end{equation}
In our analysis, we employ an indicator distortion function,

\begin{equation}
    g(S_t,W_t)=\mathbbm{1}_{\{S_t\neq W_t\}}\,.
\end{equation}
Figure~\ref{fig:aoi_vs_aoii} illustrates concurrent sample paths of both the AoI and the AoII.

\begin{figure}
    \centering
    \includegraphics[width=.9\linewidth]{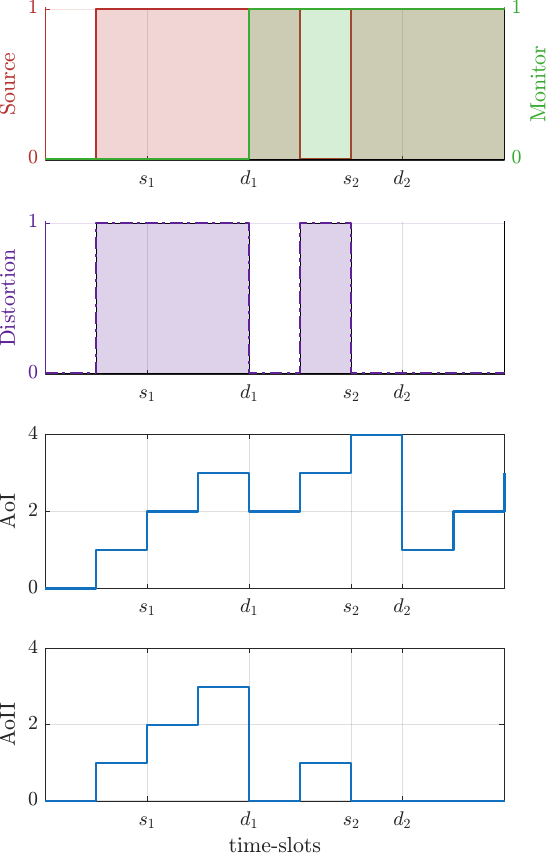}
    \caption{Concurrent discrete-time sample paths of the AoI and AoII. A binary source is sampled at slots $s_i$, $i=1,2$, and decoded by the receiver at slots $d_i$, $i=1,2$, respectively. The distortion is the absolute difference between the source and monitor values.}
    \label{fig:aoi_vs_aoii}
\end{figure}

The AoII metric has been studied primarily with discrete-time Markov sources, as they offer an analytical advantage. The source models in~\cite{9162726,11165199,10225853} are binary symmetric Markov chains. The works~\cite{bib:kbountro_harq,bib:aoii,bib:aoii_multiple_sources2,bib:kbountro_VLSF,kriouile2022minimizingageincorrectinformation} focus on arbitrarily large symmetric Markov models. In~\cite{9686027}, the source model is a bird-death Markov chain. An alternative formulation is employed in~\cite{bib:aoii_semantics2}, where the source is not modeled directly, but a binary Markov chain is used to model the distortion function. 
Those works exhibit that AoII-optimal transmission policies in resource-constrained environments commonly admit a threshold structure, a result also typical for the conventional AoI~\cite{bib:aoi_harq,bib:threshold_policy_3,bib:threshold_policy_5,bib:threshold_and_uniform_policy,bib:nonlinear_Sun}.

Generic Markov sources, i.e., without any pre-specified structure, have limited appearance in the literature. 
The seminal work~\cite{11071330} analyzes the AoII for finite and irreducible continuous-time Markov sources, proving that the optimal policy in the continuous setting is a multiple-threshold policy where thresholds depend on both the source and the estimate at the receiver. Acknowledging the intractability of the problem, a suboptimal solution is proposed in which the thresholds depend only on the estimate. Recently, a work available in~\cite{cosandal2025minimizingfunctionsageincorrect} develops a similar policy for a discrete-time Markov source. Notably, multiple-threshold policies have been previously analyzed in the setting of bird-death Markov chains in~\cite{9686027}.

Our work investigates a noisy communication system that employs a finite and irreducible discrete-time Markov source. The transmitter and receiver use a hybrid automatic repeat request (HARQ) protocol to correct transmission errors, with the decoder requesting additional transmissions to improve the decoding probability. Additionally, the transmission policy is subject to a resource constraint on the long-term average transmission rate.

The main contributions of our work are summarized as follows: 
\begin{itemize}
    \item We study AoII minimization for discrete-time Markov sources over a noisy HARQ channel under a transmission rate constraint, extending prior works that focus on structured or symmetric models.
    \item The optimization problem is formulated as a constrained Markov decision process (CMDP), and it is shown that there exists an optimal policy that is a randomized mixture of two stationary policies.
    \item To address the intractability of computing the optimal stationary policies, we develop a class of multiple-threshold transmission policies, denoted as $\mathcal{F}(S,W,r)$, where transmissions are triggered when the AoII exceeds state-dependent thresholds based on the source state, the receiver estimate, and the packet count.
    \item We establish a Markov renewal-reward structure induced by threshold policies and derive analytic expressions for the long-term average AoII and transmission rate, enabling efficient evaluation of candidate policies. Leveraging the threshold structure, we design an efficient relative value iteration (RVI) algorithm combined with a bisection search to approximate the optimal threshold policy.
    \item To accommodate scenarios with large state spaces or frequent policy refinements, we develop single-threshold policies, $\mathcal{F}(1)$, that trade optimality for efficiency by lifting the dependency on the real-time system parameters except the AoII.
\end{itemize}

\section{Problem Definition}\label{sec:problem}

\subsection{Communication Model}\label{sec:communication}

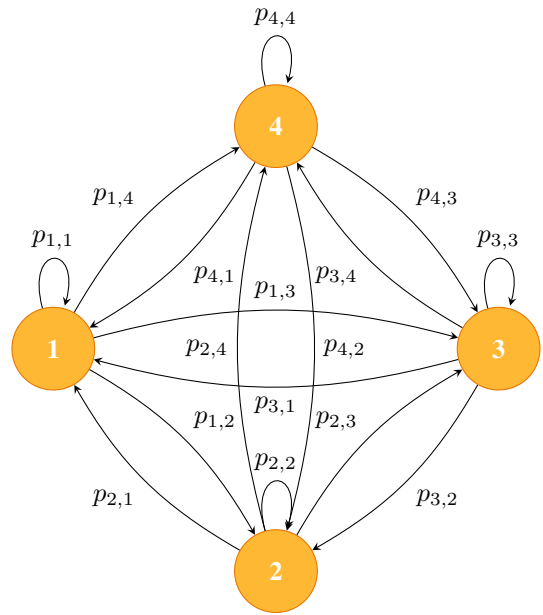
\begin{figure}
\centering
\begin{tikzpicture}[scale=0.74, transform shape=false,->, >=stealth, auto, node distance=3cm,
    every state/.style={circle, draw=orange!90!black, fill=orange!70!yellow!80, text=white, minimum size=1.1cm, font=\bfseries}]
    
    \foreach \i in {1,...,4} {
        \node[state] (q\i) at (90+360/4*\i:4cm) {\i};
    }

    \foreach \i in {1,...,4} {
        \foreach \j in {1,...,4} {
            \ifnum\i=\j
                \path (q\i) edge[loop above] node {$p_{\i,\j}$} (q\i);
            \else
                \path (q\i) edge[bend left=15] node {$p_{\i,\j}$} (q\j);
            \fi
        }
    }

\end{tikzpicture} 
\caption{A Markov chain of $N=4$ states.}
  \label{fig:source}
\end{figure}

We consider a discrete-time (slotted) communication model over a noisy channel. We focus on $N$-ary Markov sources. In this context, the (single-step) transition probability from state $i$ to state $j$ is denoted by $p_{i,j}$, $\{i,j\}\in\{1,\dots,N\}^2$. An example is depicted in Fig.~\ref{fig:source} for $N=4$. 

At each time slot, the transmitter acquires a sample from the source and decides whether to transmit or discard it. The sample to be transmitted is encoded into a packet with a channel coding scheme and communicated through a noisy channel to the receiver. Upon receiving the packet, the receiver attempts to decode it. If the decoding is successful, the receiver notifies the transmitter with an ACK feedback message. Otherwise, it sends a NACK message to request an additional transmission, and the transmitter decides whether to retransmit or reject the request.

The receiver holds the last successfully decoded value until the following successful decoding. The probability of successful decoding is specified by a non-decreasing function $d(r)\in (0,1)$, where $r$ is the number of packets at the receiver for the specific sample.

The packet transmission duration is constant and equal to one time slot, whereas the ACK/NACK packets are instantaneous. The instantaneous feedback message is a typical assumption in the literature, justified by its limited information content. 

Lastly, motivated by the need to conserve or allocate power and network resources, we impose a constraint on the transmission rate. In particular, it is required that the long-term average transmission rate, measured in packets per slot, be less than or equal to $R\in [0,1]$.

\subsection{Optimization Problem}

The general problem is formulated as an infinite-horizon average-cost constrained Markov decision problem (CMDP) denoted by $\mathcal{K}=(\mathcal{X},\mathcal{Y},P,C,R)$, where
\begin{itemize}[leftmargin=*]
    \item the state of the system at time $t$ is $X_t=(S_t,W_t,\delta_t,r_t)\in\mathcal{X}$, where $S_t\in\{1,\dots,N\}$ denotes the source state, $W_t\in\{1,\dots,N\}$ denotes the state known by the receiver, $\delta_t$ is the AoII, and $r_t\in\{0,\dots,r_{max}\}$ is the number of packets received for the current sample transmission.
    \item The actions $y_t\in \mathcal{Y}=\{0,1\}$ denote the action at time $t$, where the action space consists of the \textit{wait} ($y_t=0$) and the \textit{transmit} ($y_t=1$) actions. 
    \item The action-dependent transition probabilities $P(X_{t+1}\mid X_t, y_t)$ are summarized in Table~\ref{tab:transitions_original}, reflecting the model as described in Sec.~\ref{sec:communication}.
    \item The cost at state $X_t=(S_t,W_t,\delta_t,r_t)$ is equal to the instantaneous AoII $C(X_t)\triangleq\delta_t$.
    \item The long-term average number of ``transmit'' actions is constrained to not exceed $R\in(0,1]$.

\end{itemize}

\begin{table*}
\normalsize
\hrule
~
    \begin{equation}\label{eq:transitions}
        \begin{aligned}
        %
        %
        &P(S_{t+1},W_{t+1}=W_t,\delta_{t+1}=\delta_t+1,r_{t+1}=0\mid X_t, y_t=0) = p_{S_t,S_{t+1}} &&\forall~S_{t+1}\!\neq\! W_t\\
        &P(S_{t+1}=W_t,W_{t+1}=W_t,\delta_{t+1}=0,r_{t+1}=0\mid X_t, y_t=0) = p_{S_t,W_t}\\~\\
        %
        %
        &P(S_{t+1},W_{t+1}=W_t,\delta_{t+1}=\delta_t+1,r_{t+1}=0\mid X_t, y_t=1) = p_{S_t,S_{t+1}}\cdot (1-d(r_t)) &&\forall\ S_{t+1}\!\notin\! \{S_t,W_t\}\\
        &P(S_{t+1}=W_t,W_{t+1}=W_t,\delta_{t+1}=0,r_{t+1}=0\mid X_t, y_t=1) = p_{S_t,W_t}\cdot (1-d(r_t))\\
        &P(S_{t+1}=S_t,W_{t+1}=W_t,\delta_{t+1}=\delta_t+1,r_{t+1}=r_t+1\mid X_t, y_t=1) = p_{S_t,S_t}\cdot (1-d(r_t)) && \text{if}\ S_t\neq W_t\\
        &P(S_{t+1},W_{t+1}=S_t,\delta_{t+1}=\delta_t+1,r_{t+1}=0\mid X_t, y_t=1) = p_{S_t,S_{t+1}}\cdot d(r_t) &&\forall\ S_{t+1}\!\neq\! S_t\\
        &P(S_{t+1}=S_t,W_{t+1}=S_t,\delta_{t+1}=0,r_{t+1}=0\mid X_t, y_t=1) = p_{S_t,S_t}\cdot d(r_t)
        \end{aligned}
    \end{equation}

\hrule
\caption{The action-dependent conditional transition probabilities $P(X_{t+1}\mid X_{t},y_t)$.}
\label{tab:transitions_original}
\end{table*}

Our goal is to find the transmission policy $\psi=\{y_t\}_{t\geq 0}$ that minimizes the long-term average AoII while satisfying the transmission rate constraint. Given initial state $X_0=x$, the problem can be expressed as the following linear programming problem,

\begin{definition}[Main Optimization Problem]
\begin{equation}\label{eq:main_opt_problem}
\begin{aligned}
    \text{Minimize}\quad &\bar{J}_\psi(x) \triangleq \limsup\limits_{T\rightarrow \infty}\frac{1}{T}E_{\psi}\!\left[\sum_{t=0}^{T-1}\delta_t \mid X_0=x\right]\\
    \text{subject to}\quad &\bar{R}_\psi(x) \triangleq \limsup\limits_{T\rightarrow \infty}\frac{1}{T}E_{\psi}\!\left[\sum_{t=0}^{T-1}y_t\mid X_0=x\right]\leq R\
\end{aligned}
\end{equation}
\end{definition}

The constrained problem can be relaxed to its unconstrained Lagrangian form. Particularly, define the MDP $\mathcal{M}=(\mathcal{X},\mathcal{Y},P,C_\lambda)$, which is identical to $\mathcal{K}$, with the exception that the transmission rate constraint is replaced by a transmission penalty $\lambda$, integrated into the instantaneous cost $C_\lambda(X_t,y_t) \triangleq \delta_t+\lambda y_t$. Let $\psi_\lambda$ denote the related transmission policy.

\begin{definition}[Lagrangian Problem]
\begin{equation}\label{eq:lagrange_problem}
    \text{Minimize}\; \bar{J}_{\psi_\lambda}(X_0) {\triangleq} {\lim\limits_{T\rightarrow \infty}}\!\sup\limits_{\lambda\geq 0}\!\frac{1}{T}E_{\psi_\lambda}{\left[\sum_{t=0}^{T-1}\delta_t{+}\lambda y_t\!\mid\! X_0\right]} {-} \lambda R
\end{equation}
\end{definition}

For any fixed value of $\lambda$, let
\begin{equation}
    \psi_\lambda^* \triangleq \argmin_\psi \bar{J}_{\psi_\lambda}\,,
\end{equation}
\begin{equation}
    g_\lambda \triangleq \min_\psi \bar{J}_{\psi_\lambda}\,,
\end{equation}
denote the $\lambda$-optimal Lagrangian policy and the average cost achieved thereof, respectively.

It is easy to see that $\mathcal{M}$ is unichain, i.e., there exists a single recurrent class. Thus, the average cost $g_\lambda$ is independent of the initial state, and the optimal policy $\psi_\lambda^*$ can be found by solving the Bellman equations~\cite[Thm. 6.5.2]{bib:krishnamurthy},

\begin{equation}\label{eq:bellman_def}
    g_\lambda + V(X_t) = \min_{y_t}\{\delta_t+\lambda y_t+\!\!\sum_{X_{t+1}}\!\!P(X_{t+1}\!\mid\! X_t, y_t)V(X_{t+1})\},
\end{equation}
where $V(X_t)$ is the value function of the state $X_t$.

\section{Optimal Policy Structure}\label{sec:optimal}
This section discusses an optimal solution of the CMDP $\mathcal{K}$~\eqref{eq:main_opt_problem} and the difficulties of its computation. The main result is the following theorem.

\begin{theorem}\label{thm:cmdp_policy}
There exists an optimal policy $\psi^*$ of~\eqref{eq:main_opt_problem}, which is a randomized mixture of two stationary policies $\psi_{\lambda^+}$ and $\psi_{\lambda^-}$ that are both optimal solutions of the Lagrangian MDP~\eqref{eq:lagrange_problem} with parameters $\lambda^+$ and $\lambda^-$, respectively. In particular,
\begin{align}\label{eq:lambda_star}
        \lambda^+ &\triangleq \inf\{\lambda\in\mathbbm{R^+}:\mathbbm{E}[R(\psi_{\lambda})]\leq R\},\\
        \lambda^- &\triangleq \sup\{\lambda\in\mathbbm{R^+}:\mathbbm{E}[R(\psi_{\lambda})]\geq R\},
\end{align}
where $\mathbbm{E}[R(\psi_{\lambda})]$ is the average transmission rate achieved by $\psi_{\lambda}$.

The optimal policy $\psi^*$ randomizes between $\psi_{\lambda^+}$ and $\psi_{\lambda^-}$ whenever the induced Markov chain $\{X_t\}_{t\geq 0}$ reaches their common regeneration set

\begin{equation}\label{eq:regeneration}
    \mathcal{C}\triangleq\{(s,w,\delta,r):s=w,\delta=0,r=0\}.
\end{equation}

Upon reaching $\mathcal{C}$, the policy $\psi^*$ selects $\psi_{\lambda^-}$ with probability $\rho$ and $\psi_{\lambda^+}$ with probability $1-\rho$. The mixing probability is chosen such that the randomized policy has an average transmission rate equal to $R$, and is explicitly defined as
\begin{equation}
    \rho \triangleq \frac{R - \mathbbm{E}[R(\psi_{\lambda^+})]}{\mathbbm{E}[R(\psi_{\lambda^-})] - \mathbbm{E}[R(\psi_{\lambda^+})]}\,.
\end{equation}
\end{theorem}

\begin{proof}
    The proof is based on~\cite{bib:sennott}. The details can be found in Appendix~\ref{apdx:proof1}.
\end{proof}

The optimal policy defined in Theorem~\ref{thm:cmdp_policy} randomizes between two policies $\psi_{\lambda^+}$ and $\psi_{\lambda^-}$. The policy-induced Markov chains of those policies form distinct renewal processes that share a common regeneration set $\mathcal{C}$; analyzing this set yields the long-term average AoII of the optimal policy $\psi^*$, as follows.

For all $i\in\mathbbm{N}$, define the regeneration (stopping) times under the $\psi$-induced Markov chain,

\begin{equation}\label{eq:T_i}
    T_{i+1}(\psi)\triangleq \min\{t>T_i(\psi):S_t=W_t\mid \psi\},
\end{equation}
where $T_0(\psi)$ is assumed to be the very first visit to $\mathcal C$. Each interval $[T_i(\psi),T_{i+1}(\psi)]$ defines a \textit{regeneration cycle}.
Let
\begin{align}
    L_i(\psi)&\triangleq T_{i+1}(\psi)-T_i(\psi)\,,\label{eq:L_i}\\
    J_{i}(\psi)&\triangleq \sum_{t=T_i(\psi)}^{T_{i+1}(\psi)-1}\delta_t\,,\label{eq:J_i}
\end{align}
define the cycle length and the cumulative AoII in the $i$-th cycle, respectively, under $\psi$.

\begin{proposition}
    The long-term average AoII of $\psi^*$ equals
    \begin{equation}\label{eq:Jpsi}
        \bar{J}_{\psi^*} = \frac{\rho \mathbbm{E}[J({\psi_{\lambda^-}})]+ (1-\rho) \mathbbm{E}[J{(\psi_{\lambda^+})}]}{\rho \mathbbm{E}[L{(\psi_{\lambda^-})]}+ (1-\rho) \mathbbm{E}[L{(\psi_{\lambda^+})}]}.
    \end{equation}
\end{proposition}

\begin{proof}
Define the embedded process $\{\Psi^*_i\}_{i\geq 0}$,
\begin{equation}
    \Psi^*_i=\begin{cases}
        \psi_{\lambda^-}    &\text{w.p.}\ \rho\,,\\
        \psi_{\lambda^+}    &\text{w.p.}\ 1-\rho\,,
    \end{cases}
\end{equation}
indicating which policy is selected at the regeneration times $T_i$. The pair sequence $\left\{\left(L_i(\psi),J_i(\psi)\right)\right\}_{i\geq 0}$ constitutes a renewal-reward process. The result follows from the renewal-reward theorem.
\end{proof}

To make use of Theorem~\ref{thm:cmdp_policy}, it is necessary to be able to compute $\psi_\lambda$ and $\mathbbm{E}[R(\psi_{\lambda})]$. However, their computation is not straightforward. Particularly,
\begin{enumerate}[label=(\alph*)]
    \item The traditional approach to analyzing the Bellman equations~\eqref{eq:bellman_def} -- in the hope of deriving a threshold structure for $\psi_\lambda$ -- is very difficult due to the joint transitions of the source-receiver pair. 
    \item To compute $\mathbbm{E}[R(\psi_{\lambda})]$ in a scalable manner, the underlying MDP needs to possess a tractable structure.
\end{enumerate}

The same difficulties are inherited in the computation of $\mathbbm{E}[J(\psi_{\lambda})],~\mathbbm{E}[L(\psi_{\lambda})]$ in~\eqref{eq:Jpsi}.

Dynamic programming algorithms, such as the Relative Value Iteration (RVI) and Policy Iteration (PI)~\cite[Chap. 6]{bib:krishnamurthy}, although slow, are valuable for approximating $\lambda$-optimal policies. We present such results in Fig.~\ref{fig:lambdaRVI} that experimentally illustrate a threshold structure. Hence, although not provably optimal, the threshold structure is a sensible assumption for transmission policies.

\begin{figure}
    \centering
    \includegraphics[width=0.9\linewidth]{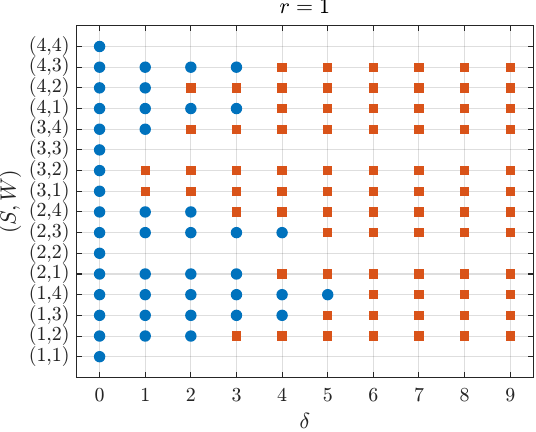}\\~\\
    \includegraphics[width=0.9\linewidth]{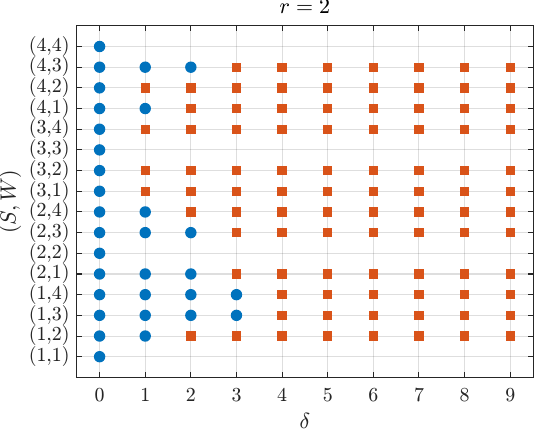}
    \caption{$\psi^*_\lambda$ policy (\ActW{} wait, \ActT{} transmit) for every source-receiver $(S,W)$ pair, and packet count $r=1$ (top) and $r=2$ (bottom). Note: When $S=W$, then $\delta=0$ always. The maximum transmission packet count is $r_\text{max}=2$, and the probability of successful decoding $d(r)$ is $0.5$ and $0.75$ for $r=1$ and $r=2$, respectively. The transmission penalty $\lambda=8$, while the source transition probabilities $p_{i,j}$ are given by the transition matrix $\left[ \begin{smallmatrix} \mathstrut
        0.52 & 0.12 & 0.18 & 0.18 \mathstrut\\ 
        0.17 & 0.57 & 0.17 & 0.09 \mathstrut\\ 
        0.03 & 0.06 & 0.72 & 0.19 \mathstrut\\ 
        0.16 & 0.10 & 0.18 & 0.56 \mathstrut
    \end{smallmatrix} \right]$.}
    \label{fig:lambdaRVI}
\end{figure}

\section{\texorpdfstring{$\mathcal{F}(S,W,r)$}{F(S,W,r)}: The Class of Full Space Multiple-Threshold Policies}

Following the analysis and our observations in Sec.~\ref{sec:optimal}, we introduce a construction method that is optimal under the assumption that $\psi_{\lambda}$ are threshold-based on $\delta$ for all $\lambda$. In this section, we study the class of stationary multiple-threshold policies $\mathcal{F}(S,W,r)$ in which the thresholds depend on the source $S$, the receiver $W$, and the packet count $r$. Formally,

\begin{definition}~\\
    $\mathcal{F}(S,W,r)$ is the class of policies, where each $\psi\in\mathcal{F}(S,W,r)$ is a function of the system state $X_t=(s,w,\delta,r)$ such that
    \begin{equation}\label{eq:y}
        \psi(s,w,\delta,r) = \begin{cases}
            0 \quad &\text{if }\delta<n(s,w,r),\\
            1 \quad &\text{if }\delta\geq n(s,w,r).
        \end{cases}
    \end{equation}
\end{definition}

Exploiting the structure of the policies in $\mathcal{F}(S,W,r)$, we can analyze the policy-induced Markov chain to compute $\bar{R}_{\psi_{\lambda}}$ and construct a fast RVI algorithm that approximates $\psi_\lambda$. Moreover, the same analysis can be used to compute $\bar{J}_{\psi^*}$~\eqref{eq:Jpsi}.

\subsection{Analysis of \texorpdfstring{$\mathcal{F}(S,W,r)$}{F(S,W,r)}}\label{sec:F(s,w,r)analysis}

Consider the Markov chain induced by $\psi\in\mathcal{F}(S,W,r)$, with the state at time $t$ denoted by $X^{(\psi)}_t=(S_t,W_t,\delta_t,r_t)\in\mathcal{X}$.

We will analyze $\{X^{(\psi)}_t\}_{t\ge 0}$ into \textit{cycles}. Each cycle starts and ends in a {regeneration} state $(S,W,\delta,r)\in\mathcal{C}$~\eqref{eq:regeneration}, i.e., every state with zero AoII. To characterize each cycle, we will adopt the definitions in~\eqref{eq:T_i},\eqref{eq:L_i},\eqref{eq:J_i} to denote the regeneration times $T_i(\psi)$, the cycle lengths $L_i(\psi)$, and the cumulative AoII $J_i(\psi)$ in the $i$-th cycle, respectively, under $\psi$. 

Notice that every state $(S,W,\delta,r)\in\mathcal{C}$ can be written as $(S,S,0,0)$, and is thereby distinguishable by the source value $S$. For all $i\in\mathbbm{N}$, define the embedded Markov chain $\{Z_i(\psi)\}_{i\ge 0}$ that tracks the regeneration states of $\{X^{(\psi)}_t\}_{t\ge 0}$,
\begin{equation}\label{eq:Z_i}
Z_i(\psi) = S_{T_i(\psi)}\,.
\end{equation}

Within each cycle, the AoII starts from zero and grows linearly until the end of the cycle. This allows writing $J_i(\psi)$ with the following alternative expression that is free of the $\delta_t$ parameter,
\begin{equation}\label{eq:J_i_alt}
    J_i(\psi) = \sum_{n=0}^{L_i(\psi)-1}n\,.
\end{equation}

The sequence $\{(Z_i(\psi),L_i(\psi),J_i(\psi))\}_{i\geq 0}$ forms a Markov renewal-reward process~\cite{cinlar2013introduction}. Let $\mathcal{Z}$ denote the state space of $\{Z_i(\psi)\}_{i\geq 0}$. For all $z\in\mathcal{Z}$, we employ the following conditional expectation notations, $\mathbbm{E}_z[J(\psi)]\triangleq \mathbbm{E}[J_i(\psi)\mid Z_i(\psi)=z]$ and $\mathbbm{E}_z[L(\psi)]\triangleq \mathbbm{E}[L_i(\psi)\mid Z_i(\psi)=z]$. In addition, let $\pi_{z\mid \psi}$ denote the steady-state probability mass function of $Z_i(\psi)$. The following corollary is an application of the renewal-reward theorem.

\begin{corollary}\label{corol:AoII}
    The long-term average AoII equals     
    \begin{equation}
        \mathbbm{E}[J(\psi)] =\frac{\sum _{z=1}^N \pi_{z\mid \psi}\mathbbm{E}_z[J(\psi)]}{\sum_{z=1}^N \pi_{z\mid \psi}\mathbbm{E}_z[L(\psi)]}.\label{eq:J}
    \end{equation}
\end{corollary}

Following the same procedure, we can characterize the long-term average transmission rate. Define the number of transmissions in the $i$-th cycle,
\begin{equation}\label{eq:barR_i}
C_i(\psi) \triangleq \sum_{t=T_i(\psi)}^{T_{i+1}(\psi)-1}y_t,\quad\forall~i~\in\mathbbm{N},
\end{equation}
and $\mathbbm{E}_z[C(\psi)]\triangleq \mathbbm{E}[C_i(\psi)\mid Z_i(\psi)=z]$ as previously. 

\begin{corollary}\label{corol:Rate}
    The long-term average transmission rate equals     
    \begin{equation}
        \mathbbm{E}[R(\psi)] =\frac{\sum_{z=1}^N \pi_{z\mid \psi}\mathbbm{E}_z[C(\psi)]}{\sum_{z=1}^N \pi_{z\mid \psi}\mathbbm{E}_z[L(\psi)]},  \label{eq:barR}  
    \end{equation}
\end{corollary}



To compute $\pi_{z\mid \psi},~\mathbbm{E}_z[J(\psi)],~\mathbbm{E}_z[C(\psi)],$ and $\mathbbm{E}_z[L(\psi)]$, we rely on the theory of absorbing Markov chains~\cite[Chap. 3]{bib:markov_chains}. Let $\{X^{(\psi)*}_t\}$ be an augmented Markov chain that mimics the transitions of $\{X^{(\psi)}_t\}$, but is also equipped with absorbing states that correspond to the regeneration states of $\{X^{(\psi)}_t\}$. The transition matrix of $\{X^{(\psi)*}_t\}$ is constructed as follows.

Denote with $\bm{Q^{(\psi)}}$ the transient-to-transient probability matrix, and with $\bm{U^{(\psi)}}$ the transient-to-absorbing probability matrix. As noted above, the ending state of the $i$-th cycle can be represented solely by the source state $S_{T_{i+1}(\psi)}\!=\!Z_{i+1}(\psi)$, simplifying the dimensionality of $\bm{U^{(\psi)}}$. Thus, $\bm{Q^{(\psi)}}_{(S,W,\delta,r),(S',W',\delta',r')}$ denotes the probability that the augmented chain transits from the transient state $(S,W,\delta,r)$ to another transient state $(S',W',\delta',r')$, while $\bm{U^{(\psi)}}_{(S,W,\delta,r),(Z)}$ denotes the transition probability to the regeneration (absorbing) state $Z\in\mathcal{Z}=\{1,\dots,N\}$. The state space of $\{X^{(\psi)*}_t\}$ is $\mathcal{X^*}=\mathcal{X}\cup\mathcal{Z}$.

Clearly, a cycle in $\{X^{(\psi)}_t\}$ corresponds to a path in $\{X^{(\psi)*}_t\}$ starting from the same state and ending in some absorbing state. It is important to note that all states in the state space of $\{X^{(\psi)}_t\}$ are considered as \emph{starting} states in the construction of $\bm{Q^{(\psi)}}$; even the regeneration states, since they are the starting points of every cycle. Concurrently, $\bm{U^{(\psi)}}$ encodes the probability to \emph{arrive} at some regeneration state. 

Notice that the AoII is unbounded, which leads to infinite-sized matrices $\bm{Q^{(\psi)}}$ and $\bm{U^{(\psi)}}$. Fortunately, the AoII dimension can be truncated without loss of information. To see this, observe that the transition probabilities do not depend explicitly on the precise value of $\delta_t$, but only whether $\delta_t$ lies below or above the corresponding threshold. The same holds for the computation of $\mathbbm{E}_z[L(\psi)],~\mathbbm{E}_z[J(\psi)],~\mathbbm{E}_z[C(\psi)]$ (cf.~\eqref{eq:L_i},~\eqref{eq:J_i_alt},~\eqref{eq:barR_i}). Therefore, the AoII dimension can be truncated at the maximum threshold value $n_{\text{max}}\triangleq \max_{\{s,w,r\}}n(s,w,r)$. More precisely, let $\tilde\delta_t$ denote the AoII index in the truncated space. Then,
\begin{equation}
    \tilde\delta_t\triangleq\min(\delta_t,n_{\text{max}}),
\end{equation}
which preserves the transition structure and all performance metrics of interest while yielding finite-dimensional matrices $\bm{Q^{(\psi)}}$ and $\bm{U^{(\psi)}}$.

With the above notes at hand, $\bm{Q^{(\psi)}}$ and $\bm{U^{(\psi)}}$ satisfy

\begin{align}
    &\bm{Q^{(\psi)}}_{(S,W,\tilde\delta,r),(S',W',\tilde\delta',r')} {=}\begin{cases}
        P(S'\!,\!W'\!,\!\tilde\delta'\!,\!r'\!\mid\! S,\!W\!,\!\tilde\delta,\!r) &\hspace{-2pt}\text{if }\tilde\delta'>0,\\
        0 &\hspace{-2pt}\text{if }\tilde\delta'=0.
    \end{cases} \\
    &\bm{U^{(\psi)}}_{(S,W,\tilde\delta,r),(Z)} =\begin{cases}
        P(Z,Z,\tilde\delta',r'\mid S,W,\tilde\delta,r) &\text{if }\tilde\delta'=0,\\
        0 &\text{if }\tilde\delta'>0 .
    \end{cases}
\end{align}
The complete transition matrix of $\{X^{(\psi)*}_t\}$ is
\begin{equation}
    \bm{P_{X^{(\psi)*}}}\triangleq\left[
\begin{array}{cc}
\bm{Q^{(\psi)}} & \bm{U^{(\psi)}} \\
\bm{0} & \bm{I}
\end{array}
\right].
\end{equation}

In the context of absorbing Markov chains, the fundamental matrix $\bm{N^{(\psi)}}\triangleq\bm{(I-Q^{(\psi)})}^{-1}$ quantifies the expected number of visits in every transient state until absorption. More precisely, $\bm{N^{(\psi)}}_{(S,W,\tilde\delta,r), (S',W',\tilde\delta',r')}$ is the expected number of visits in the state $(S',W',\tilde\delta',r')$ before absorption, when starting from the state $(S,W,\delta,r)$. In addition, the absorption probability matrix $\bm{B^{(\psi)}}\triangleq \bm{N^{(\psi)}}\cdot\bm{U^{(\psi)}}$ represents the probabilities of being absorbed into each state.

\begin{theorem}\label{thm:param}
    The statistical parameters of the $\psi$-induced system, $\psi\in\mathcal{F}(S,W,r)$, are computed as follows. 
    \begin{align}
        \mathbbm{E}_z[L(\psi)] &= \bm{m^{(\psi)}}_{(z,z,0,0)}\,,\\
        \mathbbm{E}_z[J(\psi)] &= \frac{ \bm{u^{(\psi)}}_{(z,z,0,0)}- \bm{m^{(\psi)}}_{(z,z,0,0)}}{2},\\
        \mathbbm{E}_z[C(\psi)] &= \sum_{\substack{{(S',W',\tilde\delta',r')}\,:\,\\\tilde\delta'\geq n(S',W',r')}}\bm{N^{(\psi)}}_{(z,z,0,0),(S',W',\tilde\delta',r')}\,,
    \end{align}
    where the vectors 
    \begin{align*}
        &\bm{m^{(\psi)}} = \bm{N^{(\psi)}}\bm{1},\\
        &\bm{u^{(\psi)}} = \bm{N^{(\psi)}}(\bm{1} + 2\bm{Q^{(\psi)}}\bm{N^{(\psi)}}\bm{1}),\\
        &\bm{1}\triangleq[1~1~\cdots~1]^T.
    \end{align*}
     Moreover, the steady-state mass function $\pi_{z\mid\psi}$ can be derived according to the transition probability matrix of $\{Z_i\}$, which equals
    \begin{equation*}
        P_{z,z'\mid \psi} = \bm{B^{(\psi)}}_{(z,z,0,0),(z')}\,.
    \end{equation*}
\end{theorem}

\begin{proof}
    To prove the result, we analyze $\{X^{(\psi)*}_t\}$ exploiting its absorbing structure. We derive the expressions by characterizing the time to absorption as a function of the initial state. The details can be found in Appendix~\ref{apdx:proof2}.
\end{proof}

Theorem~\ref{thm:param} provides the necessary material to compute the analytic expressions of the long-term average AoII $J_\psi$~\eqref{eq:J} and transmission rate $\bar{R}_\psi$~\eqref{eq:barR}.

We conclude this section with the following useful result.

\begin{proposition}\label{prop:y0}
    The optimal action at states with zero AoII, i.e., $(s,s,0,0)$ for all $s\in\{1,\dots,N\}$, is to wait ($y=0$).
\end{proposition}

\begin{proof}
    We utilize the RVI algorithm and show that at every iteration, the action that minimizes the value function is $y=0$. Since RVI converges to the true value function, it follows that the optimal action is also $y=0$. The proof is detailed in Appendix~\ref{apdx:proof_prop_y0}.
\end{proof}

A direct consequence of the above is the following.

\begin{corollary}
    The optimal thresholds satisfy $n(s,w,r)>0$ for all $s\in\{1,\dots,N\},w\in\{1,\dots,N\},r\in\{1,\dots,r_{\text{max}}\}$.
\end{corollary}

\subsection{Approximation of \texorpdfstring{$\mathcal{F}(S,W,r)$}{F(S,W,r)}}

The threshold structure of $\mathcal{F}(S,W,r)$ can be exploited to reduce the complexity of dynamic programming algorithms, omitting computations for the states with $\delta$ above the current threshold. It is important to note that the state space of the Lagrangian MDP is infinite because the AoII is unbounded, rendering dynamic programming algorithms intractable. To address this issue, we adopt the Approximating Sequence Method~\cite{bib:sennott_asm}. That is, truncate the state space by limiting the AoII component to at most $\Delta$. The truncated MDP is constructed by replacing the transitions $(s,w,\Delta,r)\rightarrow (s',w',\Delta+1,r')$ by $(s,w,\Delta,r)\rightarrow (s',w',\Delta,r')$. Let the resulting MDP and its optimal average cost be denoted by $\mathcal{M}_\lambda^{(\Delta)}$ and $g^{(\Delta)}_\lambda$, respectively.

\begin{theorem}\label{thm:ASM}
    Let $e^{(\Delta)}$ realize the optimal policy of $\mathcal{M}_\lambda^{(\Delta)}$. Then, \begin{enumerate}[label=\roman*)]
        \item The limit point of $e^{(\Delta)}$ is optimal for $\mathcal{M}_\lambda$.
        \item $\lim_{\Delta\rightarrow\infty}g^{(\Delta)}_\lambda=g_\lambda$.
    \end{enumerate}
\end{theorem}

\begin{proof}
    We verify Assumptions 1-2 in~\cite{bib:sennott_asm}. Then, Theorem 2.2 of the same reference is applied to prove the result. The details can be found in Appendix~\ref{apdx:ASM}.
\end{proof}

From Theorem~\ref{thm:ASM}, we infer that for sufficiently large $\Delta$, the approximation on the truncated MDP will be adequately close to the true optimal policy. An RVI algorithm that realizes this approximation is shown in Alg.~\ref{alg:RVI} in Appendix~\ref{apdx:algo}.

\subsection{Approximation of the Optimal Policy \texorpdfstring{$\psi^*$}{}}

Notably, the average transmission rate $\mathbbm{E}[R(\psi_{\lambda})]$ is decreasing with $\lambda$~\cite[Lemma 3.4]{bib:sennott}. Hence, the parameters $\lambda^-$ and $\lambda^+$, and therefore the complete optimal policy $\psi^*$, can be found via a bisection search~\cite{bib:algorithms}. The algorithm adopted in this paper is an adaptation of~\cite[Alg. 2]{9686027}. The pseudocode is given in Alg.~\ref{alg:bisection} in Appendix~\ref{apdx:algo}.

The complexity of Alg.~\ref{alg:bisection} is analyzed as follows. The size of the state space is $|\mathcal{X}|=N^2\cdot \Delta\cdot r_{max}$. Let $I$ be the number of iterations required for RVI convergence. The outer loop of Alg.~\ref{alg:bisection} is a bisection search, whose computational complexity is~$O(\log \lambda^+)$. We infer that the total worst-case complexity of Alg.~\ref{alg:bisection} is~$O(I\cdot N^2\cdot \Delta\cdot r_{max}\cdot \log \lambda^+)$.

In practice, the algorithm leverages the threshold structure to partially prune the state space evaluation. Specifically, once a threshold is found, the algorithm skips the calculation of the expected cost of waiting for all $\delta$ above the threshold.

\section{\texorpdfstring{$\mathcal{F}(1)$}{F(1)}: The Class of Single-Threshold Policies}

The previous section addressed the construction of general multiple-threshold policies that depend on all system parameters. The construction of these policies relies on an RVI algorithm that leverages the threshold structure to skip computations. However, the worst-case complexity remains high.

In this section, we introduce $\mathcal{F}(1)$, the class of single-threshold policies that apply the same threshold everywhere. Therein, the state space scales linearly. Crucially, the bisection method can directly search for the optimal threshold $n^*$ without employing RVI, drastically reducing the worst-case complexity to $O(\log n^*)$.

\begin{definition}~\\
    $\mathcal{F}(1)$ is the class of policies, where each $\psi\in\mathcal{F}(1)$ is a threshold function of the AoII $\delta$, 
    \begin{equation}
        \psi(\delta) = \begin{cases}
            0 \quad &\text{if }\delta<n,\\
            1 \quad &\text{if }\delta\geq n.
        \end{cases}
    \end{equation}
\end{definition}

Since $\mathcal{F}(1)$ is a special case of $\mathcal{F}(S,W,r)$, the results of Sec.~\ref{sec:F(s,w,r)analysis} hold. Specifically, Corollaries~\ref{corol:AoII}-\ref{corol:Rate} and Theorem~\ref{thm:param} enable the computation of the long-term average AoII and transmission rate of any policy $\psi\in\mathcal{F}(1)$. To simplify notation, let $R_{n}$ equal the long-term average transmission rate $R_{\psi}$ when $\psi\in\mathcal{F}(1)$ with threshold $n$. Algorithm~\ref{alg:bisection_n} in Appendix~\ref{apdx:algo} summarizes the search algorithm.

\section{Numerical Results}\label{sec:results}

Figure~\ref{fig:pol} illustrates the approximated policies of $\mathcal{F}(S,W,r)$ and $\mathcal{F}(1)$, for the same source model as in Fig.~\ref{fig:lambdaRVI} and transmission rate $R=0.1$.

\begin{figure}
    \centering
    \begin{minipage}[c]{0.65\linewidth}
        \centering
        \setlength{\tabcolsep}{5pt}
        \renewcommand{\arraystretch}{1.2}
        \begin{tabular}{cc}
            & $\mathcal{F}(S,W,r)$\\[4pt] 
            $S \downarrow$ &
            \begin{tabular}{|c|c|c|c|}
                \hline
                \cellcolor{orange!15!purple!85}($\Delta$) & \cellcolor{orange!45!purple!55}6 & \cellcolor{orange!25!purple!75}10 & \cellcolor{orange!25!purple!75}10 \\
                \hline
                \cellcolor{orange!35!purple!65}8 & \cellcolor{orange!15!purple!85}($\Delta$) & \cellcolor{orange!35!purple!65}8 & \cellcolor{orange!45!purple!55}6 \\
                \hline
                \cellcolor{orange!65!purple!35}2 & \cellcolor{orange!65!purple!35}2 & \cellcolor{orange!15!purple!85}($\Delta$) & \cellcolor{orange!55!purple!45}4 \\
                \hline
                \cellcolor{orange!45!purple!55}6 & \cellcolor{orange!55!purple!45}4 & \cellcolor{orange!35!purple!65}8 & \cellcolor{orange!15!purple!85}($\Delta$) \\
                \hline
            \end{tabular} \\
            & $W \rightarrow$
        \end{tabular}
    \end{minipage}
    \hfil
    \begin{minipage}[c]{0.15\linewidth}
        \centering
        \renewcommand{\arraystretch}{1.2}
        \begin{tabular}{c}
            $\mathcal{F}(1)$\\[4pt] 
            \begin{tabular}{|c|}
                \hline
                \cellcolor{orange!35!purple!65}8 \\
                \hline
                \cellcolor{orange!35!purple!65}8 \\
                \hline
                \cellcolor{orange!35!purple!65}8 \\
                \hline
                \cellcolor{orange!35!purple!65}8 \\
                \hline
            \end{tabular} \\
            \vphantom{$W \rightarrow$}
        \end{tabular}
    \end{minipage}
    \caption{The stationary threshold policy component $\psi_{\lambda^+}$ of the randomized policy $\psi^*$ for each of the threshold policy classes introduced in the previous sections. For the policy in $\mathcal{F}(S,W,r)$, only the thresholds for packet count $r=1$ are shown. The transmission rate is $R=0.1$, the maximum transmission packet count is $r_\text{max}=2$, and the probability of successful decoding $d(r)$ is $0.5$ and $0.75$ for $r=1$ and $r=2$, respectively. The source transition probabilities $p_{i,j}$ are given by the transition matrix $\left[ \begin{smallmatrix} \mathstrut
        0.52 & 0.12 & 0.18 & 0.18 \mathstrut\\ 
        0.17 & 0.57 & 0.17 & 0.09 \mathstrut\\ 
        0.03 & 0.06 & 0.72 & 0.19 \mathstrut\\ 
        0.16 & 0.10 & 0.18 & 0.56 \mathstrut
    \end{smallmatrix} \right]$.}
    \label{fig:pol}
\end{figure}

\begin{figure}
    \centering
    \includegraphics[width=0.93\linewidth]{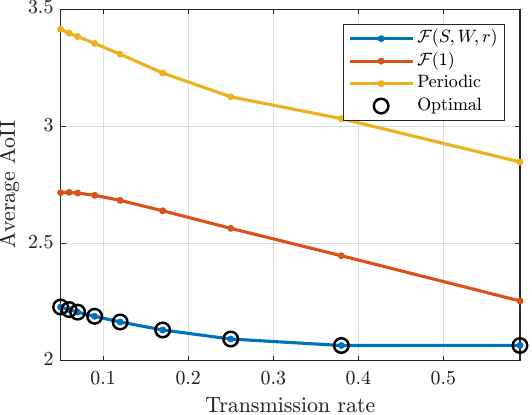}
    \caption{Average AoII versus transmission rate for a random source - $N=4$}
    \label{fig:aoii1}
    ~\\
    \includegraphics[width=0.93\linewidth]{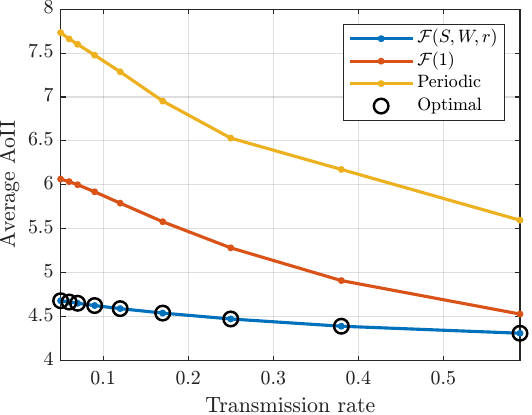}
    \caption{Average AoII versus transmission rate for a random source - $N=8$}
    \label{fig:aoii2}
    ~\\
    \includegraphics[width=0.93\linewidth]{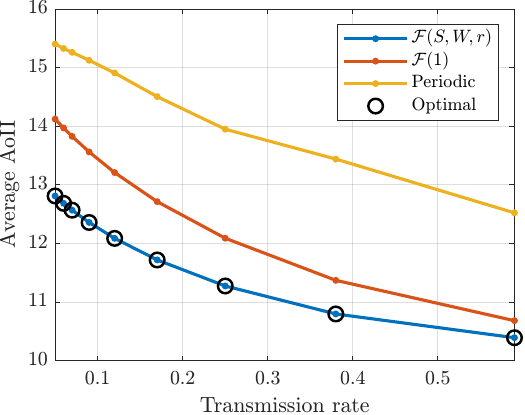}
    \caption{Average AoII versus transmission rate for a random source - $N=16$}
    \label{fig:aoii3}
\end{figure}

Hereafter, each experiment is carried out on a distinct, randomly drawn Markov source. For the experiments to be informative, we bias the diagonal elements of $P$ to be the largest element in its row. Otherwise, there is a high probability that the selected Markov source contains multiple states where transmissions are never preferable, making the comparison across the different source models complex. In practice, this assumption implies that the duration of a transmission (one time slot) is shorter than the average time required for the source to change value. In our experiments, each row possesses $N$ elements that are uniformly distributed and normalized, and the diagonal element is assigned the largest of those elements.

Similar to~\cite{bib:aoi_harq} and~\cite{bib:kbountro_harq}, the probability of failed decoding is modeled as an exponentially decreasing function. Thus, the probability of successful decoding is written as
\begin{equation}
    d(r) = 1 - p_e c^r\,,\quad \text{for $\ 0\leq r\leq r_{max}$}\,,
\end{equation}
where $p_e$ is the initial error rate and $c\in(0,1]$ is the decaying error rate constant. To facilitate comparison of the transmission policies, the chosen parameters are constant at $c=0.5$, $p_e=0.5$, and $r_{max}=2$.

Figures~\ref{fig:aoii1}~-~\ref{fig:aoii3} illustrate the average AoII achieved by the proposed policies versus the transmission rate, together with the AoII of the optimal policy without the constraint of the threshold structure. Additionally, results are drawn for a periodic policy that transmits every $\lceil 1/R \rceil$ time slots. Each figure is dedicated to a different cardinality, \mbox{$N=4,8,16$}. 

Our results confirm that searching within the $\mathcal{F}(S,W,r)$ class yields nearly optimal policies, as their performance matches the globally optimal policies in all tested instances. This approach yields strict improvements over all alternatives, especially in the lower transmission-rate region. Finally, we observe that even the elementary single-threshold policy $\mathcal{F}(1)$ dramatically outperforms periodic transmissions.

\section{Conclusion}

This work investigates transmission policies for minimizing the AoII metric of discrete-time Markov sources. The considered communication model employs an HARQ protocol for error correction under a constrained transmission rate. We demonstrate that the optimization problem can be formulated as a CMDP, and that there exists an optimal policy which is a randomized mixture of two stationary policies. We propose a method for deriving feasible multiple-threshold policies, where transmissions occur as long as the AoII is larger than state-dependent thresholds that rely on the source, receiver, and packet count. The resulting system is analyzed as a Markov renewal-reward process, allowing for the derivation of analytic expressions for both the long-term average AoII and the average transmission rate. The threshold allocation is carried out via a relative value iteration algorithm acting on a truncated state space, paired with a bisection search to approximate the optimal transmission penalty. For reduced complexity, a simpler single-threshold policy can be derived that significantly scales down the state space. The results demonstrate the superiority of threshold-based policies over conventional periodic transmission schemes and highlight the potential of multiple-threshold policies in AoII optimization problems.

The establishment of provably optimal policies remains an open problem, even though there is strong evidence that the multiple-threshold policy developed herein is the globally optimal solution. Another interesting research direction lies in the decentralized control of multiple discrete sources sharing the same communication channel.

\bibliographystyle{IEEEtran}
\bibliography{IEEEabrv,biblio}

\clearpage

\appendix

\subsection{Proof of Theorem~\ref{thm:cmdp_policy}}\label{apdx:proof1}

\begin{proof}
First, we prove that assumptions 1-5 of~\cite{bib:sennott} hold for our problem. Under those properties, the results of~\cite{bib:sennott} yield the existence and structure of an optimal policy. 

\begin{definition}[Definition 2.3 of~\cite{bib:sennott}]
Let $G\subset \mathcal{X}$ be a non-empty subset of states of the CMDP. Given a state $x\in\mathcal{X}$, let $\mathcal{R}(x,G)$ be the class of policies $\psi$ such that $P_{\psi}(X_t\in G \text{ for some } t\geq 1\mid X_0 = x)=1$ and the expected time $T_{x,G}$ of the first passage from $x$ to $G$ under $\psi$ is finite. Let $\mathcal{R^*}(x,G)$ be the class of policies $\psi\in \mathcal{R}(x,G)$ such that, in addition, the expected average AoII $\bar{J}_{x,G}(\psi)$ and the expected transmission rate $\bar{R}_{x,G}(\psi)$ of the first passage from $x$ to $G$ are finite.
\end{definition}

Hereafter, we state and prove the assumptions mentioned above for our problem.

\begin{assumption}
For all $b\!>\!0$, the set $G(b)\triangleq \{(s,w,\delta,r)\mid \text{ there exists an action } y \text{ such that } \delta+y\leq b\}$ is finite.
\end{assumption}
\begin{proof}
Since $\delta\in\mathbbm{N}$, the states with $\delta\leq b$ are always finite, which implies that $G(b)$ is also finite.
\end{proof}

\begin{assumption}There exists a stationary policy $e$ that induces a Markov chain with the following properties: the state space $\mathcal{X}$ consists of a single (non-empty) positive recurrent class $\Lambda$ and a set $U$ of transient states such that $e\in\mathcal{R^*}(u,\Lambda)$, for any $u\in U$. Moreover, both the average AoII $\bar{J}_e$ and the average transmission rate $\bar{R}_e$ on $\Lambda$ are finite.
\end{assumption}
\begin{proof}

Consider the policy $e(x)=0$ for all $x\in\mathcal{X}$, i.e., the always-wait policy. Let $X_0=(s_0,w_0,\delta_0,r_0)$. It can be seen that the policy-induced Markov chain $\{X_e(t)\}=\{s_t,w_t,\delta_t,r_t\}$, $t\in\mathbbm{N}$ has a single positive recurrent class $\Lambda_e=\{(s,w,\delta,r):s\in\{1,\dots,N\},w=w_0,\delta\in\mathbbm{N},r=0\}$. The transient set is $U=\{(s,w,\delta,r):s\in\{1,\dots,N\},w=w_0,\delta\in\mathbbm{N},r>0\}$, i.e., they differentiate only on the $r$ component. Recall that the wait action always reverts $r$ to zero, which implies that the expected time of the first passage $T_{u,\Lambda}\leq 1$. In addition, the expected transmission rate of the first passage $\bar{R}_{u,\Lambda}(e)=0$, while the expected average AoII of the first passage $\bar{J}_{u,\Lambda}(e)$ is trivially finite. Hence, $e\in\mathcal{R^*}(u,\Lambda)$.

In addition, the average transmission rate $\bar{R}_{e} = 0$. To show that $\bar{J}_{e}$ is finite, notice that the AoII $\delta_t$ reverts to zero when $s_t=w_0$, and grows linearly otherwise. We infer that the Markov chain forms a renewal process wherein the regeneration (stopping) times are $\chi_{i+1}\triangleq \inf\{t>\chi_i:s_t=w_0\}$ and the cycle lengths $\mu_i\triangleq \chi_{i+1}-\chi_i$. Therefore, the cumulative AoII in the $i$-th cycle equals 

\begin{equation}
    g_i=\sum_{n=0}^{\mu_i-1}n=\frac{\mu_i^2-\mu_i}{2}\,.
\end{equation}
Hence,
\begin{equation}
    \mathbbm{E}[g]=\frac{\mathbbm{E}[\mu^2]-\mathbbm{E}[\mu]}{2}\,.
\end{equation}
By the renewal-reward theorem,
\begin{equation}
    \bar{J}_e(X_0)=\frac{\mathbbm{E}[g]}{\mathbbm{E}[\mu]} = \frac{\mathbbm{E}[\mu^2]-\mathbbm{E}[\mu]}{2\mathbbm{E}[\mu]}\,.
\end{equation}
Recall that the source Markov chain $\{S_t\}$ is finite and irreducible, which implies that $\mathbbm{E}[\mu^2]<\infty$. Therefore, $\bar{J}_e(X_0)<\infty$.
\end{proof}

\begin{assumption}Given any two states $x\neq x'\in \mathcal{X}$, there exists a policy $\psi$ (a function of $x$ and $x'$) such that $\psi \in \mathcal{R^*}(x,x')$.
\end{assumption}
\begin{proof}
Again, consider the policy $\psi(x)\!=\!1$ for all $x\in\mathcal{X}$, and recall that the induced Markov chain consists of a single positive recurrent class that is exactly the state space $\mathcal{X}$. Thus, any two states communicate. The average AoII and the transmission rate of the first passage are trivially finite, i.e., $\psi \in \mathcal{R^*}(x,x')$.
\end{proof}

\begin{assumption}
If a deterministic policy $f$ has at least one positive recurrent state, then it has a single positive recurrent class $\Lambda$. Moreover, if $X_0\notin \Lambda$, then $\psi\in\mathcal{R^*}(X_0,\Lambda)$.
\end{assumption}
\begin{proof}
If $f$ transmits for all source states, then there exists a single positive recurrent class $\Lambda=\mathcal{X}$ and the state $X_0$ is always in $\Lambda$. Assume that $f$ never transmits for some subset of the source states $S_{I}\subset\{1,\dots,N\}$. Denote the complement set with $S_{T}$. Then, the induced Markov chain consists of a single positive recurrent $\Lambda=\{(s,w,\delta,r):s\in\{1,\dots,N\},~w\in S_{T},~\delta\in\mathbbm{N},~r\in\{1,\dots,r_{max}\}\}$. Let $X_0=(s_0,w_0,\delta_0,r_0)\notin \Lambda$. The source state $s_0$ communicates with $S_T$. Therefore, $f$ is bound to transmit some source state from $S_T$ in finite time, and the receiver is bound to decode one in finite time. At this time, the chain reaches the class $\Lambda$.
\end{proof}

\begin{assumption}
There exists a policy $\psi$ such that $\bar{J}_\psi(X_0) \!<\! \infty$ and $\bar{R}_\psi(X_0) \!<\! R$.
\end{assumption}
\begin{proof}
This is shown in the proof of Assumption 3.
\end{proof}

Having proven assumptions 1-5, the results~\cite[Thm. 2.5, Prop. 3.2, Lemmas 3.4, 3.7, 3.9, 3.10, 3.12]{bib:sennott} also hold for our CMDP. 

By definition, $\lambda^-$ and $\lambda^+$ converge to the same $\lambda^*$ from below and above, respectively. By~\cite[Lemma 3.4, Lemma 3.7]{bib:sennott}, the policies $\psi_{\lambda^-}$ and $\psi_{\lambda^+}$ also converge and have the same cost $g_{\lambda^*}$. Since $\mathcal{C}$~\eqref{eq:regeneration} is positive recurrent for both policies, a randomized policy that randomizes at $\mathcal{C}$ is also $\lambda^*$-optimal~\cite[Lemma 3.9]{bib:sennott}. Finally, applying~\cite[Lemmas 3.10, 3.12]{bib:sennott}, the optimality of the randomized policy $\psi^*$ for the CMDP is deduced.

\end{proof}

\subsection{Proof of Theorem~\ref{thm:param}}\label{apdx:proof2}
\begin{proof}

    To simplify the notation throughout this proof, we omit the policy $\psi$ from the notation. All quantities below are understood to depend on $\psi$.

    Let $\bm{\tau}_{X^*_t}$ denote the time until absorption starting from the state $X^*_t$. More precisely,

    \begin{equation}
        \tau_{X^*_t} \triangleq \min\{n\geq t:X^*_n\in \mathcal{Z}, n>t\}-t\,.
    \end{equation}
    Clearly, 
    \begin{equation}
        L_i=\tau_{X^*_{T_i}}, \quad\forall~i\in\mathbbm{N}.
    \end{equation}

    Assume that $X^*_t$ is a transient state, as is the case with cycle-starting states. First-step conditioning of $\tau{(X^*_t)}$ yields

    \begin{align}
        \bm{\tau}_{X^*_t} &= 1 + \bm{\tau}_{X^*_{t+1}}\label{eq:tau}\\
        \Rightarrow\bm{\tau}_{X^*_{t}}^2 &= 1 + 2\bm{\tau}_{X^*_{t+1}} + \bm{\tau}_{X^*_{t+1}}^2.\label{eq:tauSq}
    \end{align}

    Let $\mathbbm{E}_{x^*}[\bm{\tau}]\triangleq \mathbbm{E}[\bm{\tau}_{X^*_{t}}\mid X^*_t=x^*]$ and $\mathbbm{E}_{x^*}[\bm{\tau}^2]\triangleq \mathbbm{E}[\bm{\tau}^2_{X^*_{t}}\mid X^*_t=x^*]$, $x^*\in\mathcal{X^*}$. Taking the conditional expectation on either sides of~\eqref{eq:tau} and~\eqref{eq:tauSq} returns the following.
    
    \begin{equation}\label{eq:Etau}
    \begin{aligned}
        \mathbbm{E}[\bm{\tau}_{X^*_{t}}\mid X^*_t=x^*] &=  1 + \mathbbm{E}[\bm{\tau}_{X^*_{t+1}}\mid X^*_t=x^*] \\
        \Rightarrow \mathbbm{E}_{x^*}[\bm{\tau}] &= 1 + \sum_{x'\in\mathcal{X}}\bm{Q}_{x^*,x'}\mathbbm{E}_{x'}[\bm{\tau}],
    \end{aligned}
    \end{equation}

    \begin{equation}\label{eq:EtauSq}
    \begin{aligned}
        \mathbbm{E}[\bm{\tau}_{X^*_{t}}^2\mid X^*_t=x^*] &= 1 + 2\mathbbm{E}[\bm{\tau}_{X^*_{t+1}}\mid X^*_t=x^*] \\
                                                        &\phantom{=}+ \mathbbm{E}[\bm{\tau}_{X^*_{t+1}}^2\mid X_t^*=x^*] \\
        \Rightarrow \mathbbm{E}_{x^*}[\bm{\tau}^2] &=  1 + 2\sum_{x'\in\mathcal{X}}\bm{Q}_{x^*,x'}\mathbbm{E}_{x^*}[\bm{\tau}] \\
                                                    &\phantom{=}+ \sum_{x'\in\mathcal{X}}\bm{Q}_{x^*,x'}\mathbbm{E}_{x^*}[\bm{\tau}^2].
    \end{aligned}
    \end{equation}
    
    Define the column vectors $\bm{m}$ and $\bm{v}$ with elements $\bm{m}_{x^*}\triangleq \mathbbm{E}_{x^*}[\bm{\tau}]$ and $\bm{v}_{x^*}\triangleq \mathbbm{E}_{x^*}[\bm{\tau}^2]$, respectively. In addition, define the column vector $\bm{1}\triangleq[1~1~\cdots~1]^T$. Then,~\eqref{eq:Etau} and~\eqref{eq:EtauSq} can be written more concisely in matrix-vector form as,

    \begin{equation}\label{eq:m}
    \begin{aligned}
        \bm{m} &= \bm{1} + \bm{Qm}\\
        \Rightarrow\bm{m} &= (\bm{I}-\bm{Q})^{-1}\bm{1}\\
        \Rightarrow\bm{m} &= \bm{N}\bm{1},
    \end{aligned}
    \end{equation}
    
    \begin{equation}
    \begin{aligned}
        \bm{u} &= \bm{1} + 2\bm{Qm} + \bm{Qu}\\
        \Rightarrow \bm{u} &= (\bm{I}-\bm{Q})^{-1}(\bm{1} + 2\bm{Q}\bm{m})\\
        \Rightarrow \bm{u} &= \bm{N}(\bm{1} + 2\bm{Q}\bm{N}\bm{1}).
    \end{aligned}
    \end{equation}

    First, it can be seen that
    \begin{equation}
        \mathbbm{E}_z[L] = \bm{m}_{(z,z,0,0)}.    
    \end{equation}
    
    Moreover, 
    \begin{equation}
    \begin{aligned}
        \mathbbm{E}_z[J] &= \mathbbm{E}\left[\left.\sum_{k=0}^{T_{i+1}-T_{i}-1}k~\right\vert~ Z_i=z\right] \\
        &= \mathbbm{E}\left[\left.\frac{(T_{i+1}-T_i)^2-(T_{i+1}-T_i)}{2}~\right\vert~ Z_i=z\right]\\
        &= \mathbbm{E}\left[\left.\frac{L_i^2-L_i}{2}~\right\vert~ Z_i=z\right]\\
        &= \frac{ \mathbbm{E}[L_i^2\mid Z_i=z]- \mathbbm{E}[L_i\mid Z_i=z]}{2}\\
        &= \frac{ \mathbbm{E}_{(z,z,0,0)}[\bm{\tau}^2]- \mathbbm{E}_{(z,z,0,0)}[\bm{\tau}]}{2}\\
        &=  \frac{ \bm{u}_{(z,z,0,0)}- \bm{m}_{(z,z,0,0)}}{2}.
    \end{aligned}
    \end{equation}

    To compute $\mathbbm{E}_z[C]$, we invoke the expectation of~\eqref{eq:y} over the visited states until absorption, starting from the state $(z,z,0,0)$. Since $\bm{N}_{(z,z,0,0),(S',W',\tilde\delta',r')}$ is the average number of visits in $(S',W',\tilde\delta',r')$ before absorption, it follows that
    
    \begin{equation}\label{eq:active}
    \begin{aligned}
        \mathbbm{E}_z[C]&=\mathbbm{E}\left[\left.\sum_{t=T_i}^{T_{i+1}-1}y_t~\right\vert~ Z_i=z\right]\\
        &=\mathbbm{E}\left[\sum_{{(S',W',\tilde\delta',r')}}\psi(S',W',\tilde\delta',r')\bm{N}_{(z,z,0,0),(S',W',\tilde\delta',r')}\right]\\
        &=\sum_{{(S',W',\tilde\delta',r')}}\mathbbm{E}\left[\psi(S',W',\tilde\delta',r')\right]\bm{N}_{(z,z,0,0),(S',W',\tilde\delta',r')}\\
        &=\sum_{{(S',W',\tilde\delta',r')}}\hspace{-5.7pt}P\!\left(\psi(S',W',\tilde\delta',r')\!=\!1\right)\!\bm{N}_{(z,z,0,0),(S',W',\tilde\delta',r')}\\
        &= \sum_{\substack{{(S',W',\tilde\delta',r')}\,:\,\\\tilde\delta'\geq n(S',W',r')}}\bm{N}_{(z,z,0,0),(S',W',\tilde\delta',r')}\,.
    \end{aligned}
    \end{equation}
    
    Now, assume that $Z_{{i+1}}=z$ is the ending (absorbing) state of the $i$-th cycle. Then, the starting state of the next is $X_{T_{i+1}}=(z,z,0,0)$. Let $P^Z$ denote the transition matrix of $\{Z_i\}$. It follows that $P^Z$ can be constructed as

    \begin{equation}
        P^Z_{z,z'} = \bm{B}_{(z,z,0,0),(z')}\,.
    \end{equation}

\end{proof}

\subsection{Proof of Theorem~\ref{thm:ASM}}\label{apdx:ASM}
\begin{proof}
    The result follows from~\cite[Thm. 2.2]{bib:sennott_asm}, by which it sufficed to show that the assumptions~\cite[ASM 1, ASM 2]{bib:sennott_asm} hold for our problem. Before proving the assumptions, we define some necessary quantities.

    For this proof, the quantities in $\mathcal{M}^{(\Delta)}$ are superscripted by $\Delta$, and the components of $x$ are denoted by $x^{(s)},x^{(w)},x^{(\delta)},x^{(r)}$.
    
    Let $\theta$ be an arbitrary policy. For a discount factor $a$, the total expected discounted cost under $\theta$ is given by
    \begin{equation}
        V_{\theta,a}(x)=\mathbbm{E}_\theta\left[\left.\sum_{t=0}^\infty a^t(\delta_t+\lambda y_t)~\right\vert~X_0=x\right].
    \end{equation}
    The value function $V_{a}(\cdot)=\inf_{\theta}V_{\theta,a}(x)$ is the best that can be achieved. The minimum expected discounted cost for operating the system from time $t = 0$ to $t = n-1$ is given by
    \begin{equation}
        v_{a,n}(x)=\inf_\theta\mathbbm{E}_\theta\left[\left.\sum_{t=0}^{n-1} a^t(\delta_t+\lambda y_t)~\right\vert~X_0=x\right],
    \end{equation}
    which is also written in recursive form as
    \begin{equation}\label{eq:v_recurs}
    \begin{aligned}
        v_{a,n+1}(x)=\min_y\Bigg\{&x^{(\delta)}+\lambda y\\&+ a\sum_{x'\in\mathcal{X}} P(x'\mid x,y)v_{a,n}(x')\Bigg\},
    \end{aligned}
    \end{equation}
    
    Let $h_a^\Delta(x)= V_a^\Delta(x)- V_a^\Delta(x_{ref})$ be the relative value function in $\mathcal{M}^{(\Delta)}$, where $x_{ref}\in\mathcal{X}^\Delta$ is some chosen reference state.

    Given a randomized stationary policy $\psi$ of $\mathcal{M}$, the $\Delta$-restriction of $\psi$ is the randomized stationary policy $\psi\vert\Delta$ of $\mathcal{M}^{(\Delta)}$ that chooses actions under the same distributions as $\psi$. Let $m_{x,x_{ref}}(\psi\vert\Delta)$ and $c_{x,x_{ref}}(\psi\vert\Delta)$ denote the expected time and expected cost, respectively, of the first passage from $x$ to $x_{ref}$.

    Given initial state $X_0=x$, the average cost under a policy $\theta$ is given by
    \begin{equation}
        J_{\theta}(x)=\limsup_{T\rightarrow \infty}\frac{1}{T}\mathbbm{E}_\theta\left[\left.\sum_{t=0}^{T-1} \delta_t+\lambda y_t~\right\vert~X_0=x\right].
    \end{equation}
    Then $J(\cdot) = \inf_\theta J_\theta(\cdot)$ is the best that can be achieved.

    The structure of our truncated MDP matches that of \textit{augmentation type approximating sequences}~\cite[Def. 3.1]{bib:sennott_asm}. Specifically, let $P(x'\mid x,y~;\Delta)$ denote the action-dependent transition probabilities in $\mathcal{M}^{(\Delta)}$. Assume that $P(\mu\mid x,y)>0$ for some $\mu\in\mathcal{X}^\Delta$ and $\mu\notin\mathcal{X}^\Delta$. Then, this excess probability is redistributed to some marginal states in $\mathcal{X}^\Delta$. In other words, there exists a probability distribution $q_{x'}(x,\mu,y,\Delta)$, such that

    \begin{equation}
    \begin{aligned}
        P(x'\mid x,y~;\Delta) &=
            P(x'\mid x,y) \\ &\phantom{=}+ \sum_{\mu\notin\mathcal{X}^\Delta} P(\mu\mid x,y)q_{x'}(x,\mu,y,\Delta)\,.
    \end{aligned}
    \end{equation}
    In our specific model, 
    \begin{equation}
        q_{x'}(x,\mu,y,\Delta)=\begin{cases}1 & \text{if }  x'=\{\mu^{(s)},\mu^{(w)},\Delta,\mu^{(r)}\}\\
        &\phantom{\text{if }}\text{and } \mu^{(\delta)}=\Delta+1,\\
            0 & \text{otherwise.}
        \end{cases}
    \end{equation}
    
    \begin{assumption}[ASM 1 of~\cite{bib:sennott_asm}]\label{asm1}
        There exist a nonnegative (finite) constant $L$, a nonnegative (finite) function $M(\cdot)$ on $\mathcal{X}$, and constants $\Delta_0$ and $a\in[0,1)$, such that $-L < h_a^\Delta(x) < M(x)$, for all $x\in\mathcal{X}^\Delta,~\Delta\ge \Delta_0,~a\in(a_0,1)$.
    \end{assumption}
    \begin{proof}
        First, we specify the constant $L$.
        \begin{equation}
            \begin{aligned}
            V_a^\Delta(x) &= \inf_{\theta}\mathbbm{E}_\theta\left[\left.\sum_{t=0}^\infty a^t(\delta_t+\lambda y_t)~\right\vert~X_0=x\right]\\
                            &\leq \sum_{t=0}^\infty a^t(x^{(\delta)} + t+\lambda).
            \end{aligned}
        \end{equation}
        It can be seen that the series is finite because $a\in[0,1)$, and hence $V_a^\Delta(x)<\infty~\forall~x\in\mathcal{X}^\Delta$. Moreover, $V_a^\Delta(x)\geq 0$ because it is a sum of non-negative terms. Therefore,
        \begin{equation}
        \begin{aligned}
            & h_a^\Delta(x) -V_a^\Delta(x) =- V_a^\Delta(x_{ref})\\
            \Rightarrow & h_a^\Delta(x) \geq- V_a^\Delta(x_{ref})
        \end{aligned}
        \end{equation}
        Hence, we can set $L=\max_x V_a^\Delta(x)$.
        
        To specify the upper bound $M(\cdot)$, by~\cite[Prop. 4.2]{bib:sennott_asm}, it suffices to find a randomized stationary policy $\psi$ of $\mathcal{M}$ and an $\Delta_0$ such that
        \begin{enumerate}[label=\roman*)]
            \item For $\Delta \geq \Delta_0$ and $x \in \mathcal{X}^\Delta - {x_{ref}}$, there exists a $\psi\vert\Delta$ path from $x$ to $x_{ref}$ in $\mathcal{X}^\Delta$.
            \item The sequence $\{c_{x,x_{ref}}(\psi\vert\Delta)\}_{\Delta>\Delta_0}$, is bounded, for all $x \neq x_{ref}$.
        \end{enumerate}

        Define the time of the first passage from $x$ to $x_{ref}$, $\tau_{x,x_{ref}}\triangleq \min\{t\geq 0: X_t^\Delta=x_{ref}\mid X_0^\Delta = x\}$. Let $\psi$ be the always-transmit policy. The $\psi$-induced Markov chain is irreducible, and so there exists a $\psi\vert\Delta$ path from $x$ to $x_{ref}$. For the expected cost, we have
        \begin{equation}
            \begin{aligned}
                c_{x,x_{ref}}(\psi\vert\Delta) &=\mathbbm{E}\left[\left.\sum_{t=0}^{\tau_{x,x_{ref}}} a^t(\delta_t^\Delta+\lambda)~\right\vert~X_0=x\right]\\
                &\leq \sum_{t=0}^{\tau_{x,x_{ref}}} a^t(x^{(\delta)}+t+\lambda),
            \end{aligned}
        \end{equation}
        where $x^{(\delta)}$ is the AoII component of $x$. It can be seen that the series is finite because $a\in[0,1)$. Therefore, the upper bound $M(\cdot)$ exists. Specifically, $M(x)=\sum_{t=0}^{\tau_{x,x_{ref}}} a^t(x^{(\delta)}+t+\lambda)$.
    \end{proof}

    Under Assumption~\ref{asm1}, it can be shown the optimal average cost in $\mathcal{M}^\Delta$, for $\Delta>\Delta_0$, is $J^\Delta=\lim_{a\rightarrow 1}(1-a)V_a^\Delta(x)$ for all $x\in\mathcal{X}^\Delta$.
    \begin{assumption}[ASM 2 of~\cite{bib:sennott_asm}]
        We have $J^*\triangleq \limsup J^\Delta <\infty$ and $J^* \leq J(x)$, for all $x\in\mathcal{X}$.
    \end{assumption}
    \begin{proof}
        By~\cite[Corollary 5.2]{bib:sennott_asm}, it suffices to show that there exist $a\in[0,1)$ and $\Delta_0$ such that, for all $x\in\mathcal{X}^\Delta,~\mu\notin\mathcal{X}^\Delta~,y\in\mathcal{Y},~\Delta\ge \Delta_0,~a\in(a_0,1)$, and $n\geq 0$, we have 
        \begin{equation}\label{eq:asm2_cor52}
            \sum_{x'\in\mathcal{X}^\Delta} q_{x'}(x,\mu,y,\Delta)v_{a,n}(x')\leq v_{a,n}(\mu).
        \end{equation}
        Notice that
        \begin{equation}
            \sum_{x'\in\mathcal{X}^\Delta}\!\! q_{x'}(x,\mu,y,\Delta)v_{a,n}(x')=v_{a,n}((\mu^{(s)},\mu^{(w)},\Delta,\mu^{(r)})).
        \end{equation}
        Thus, to prove the inequality~\eqref{eq:asm2_cor52}, it suffices to show that $v_{a,n}((\mu^{(s)},\mu^{(w)},\Delta,\mu^{(r)}))\leq v_{a,n}((\mu^{(s)},\mu^{(w)},\mu^{(\delta)},\mu^{(r)}))$ for all $\mu^{(\delta)}>\Delta$. The inequality can be proved by induction on $n$ in~\eqref{eq:v_recurs}. Specifically, for $n=0$, it holds trivially since $v_{a,0}(\cdot)=0$. Expanding the summation in~\eqref{eq:v_recurs}, it can be shown that if the inequality holds for $n$, then it also holds for $n+1$. The details are omitted for brevity.
    \end{proof}
\end{proof}

\subsection{Proof of Proposition~\ref{prop:y0}}\label{apdx:proof_prop_y0}

\begin{proof}
    We utilize the RVI algorithm and show that at every iteration, the action that minimizes the value function is $y=0$. Let $V_t(\cdot,\cdot)$ denote the estimation of the value function at iteration $t\in\mathbbm{N}$. Next, define the (exact) Bellman operator,
    
    \begin{equation}\label{eq:bellman_op}
        TV_t(x) \triangleq \min_{y}\{\delta+\lambda y + \sum_{x'}P(x'\mid x, y)V_t(x')\}.
    \end{equation}

    Let $x_{ref}$ be some arbitrary reference state. The RVI updates its estimate as follows:
    \begin{equation}\label{eq:rvi_update}
        V_{t+1}(x) = TV_t(x) - TV_t(x_{ref}),\quad t=1,2,3,\dots
    \end{equation}
    Let $x=(s,s,0,0)$. Substituting the transition probabilities from~\eqref{eq:transitions} into~\eqref{eq:bellman_op}, and after minor algebraic operations, we obtain

    \begin{align}
        TV_t(x) = &\min\{\delta \!+\! \sum_{s'\neq s}p_{s,s'}V_t(s',s,1,0) \!+\! p_{s,s}V_t(s,s,0,0),\nonumber \\
        &\delta \!+\! \lambda \!+\! \sum_{s'\neq s}p_{s,s'}V_t(s',s,1,0) \!+\! p_{s,s}V_t(s,s,0,0)\},
    \end{align}
    where the first part in the minimum operator corresponds to $y=0$ and the second part to $y=1$. It is trivial to see that the action that achieves the minimum is always $y=0$. Since this holds for all $t\in\mathbbm{N}$ and the RVI converges to the true value function, it follows that the optimal action is also $y=0$.
\end{proof}

\subsection{Algorithms}\label{apdx:algo}

\begin{algorithm}
\caption{Threshold-based RVI}
\label{alg:RVI}
\begin{algorithmic}[1]
\STATE \textbf{Inputs:} Transmission cost $\lambda$, maximum AoII $\delta_{\max}$
\STATE Initialize value function $V(x) \gets 0$ for all $x\in\mathcal{X}$
\STATE Initialize thresholds $n(s,w,r)\gets \Delta$ for all $(s,w,r)$
\STATE Choose reference state $x_{ref}=(s_0,w_0,\delta_0,r_0)$

\REPEAT
    \FORALL{$(s, w, r)$}
        \FOR{$\delta = 0$ \TO $\Delta$}
            \STATE $x\gets (s, w, \delta, r)$
            
            \STATE Compute action values: \label{line:V}
            \[
            \begin{aligned}
            V^0(x) &= \delta + \sum_{x'} P(x'\mid x,y=0) V(x')\\
            V^1(x) &= \delta + \lambda + \sum_{x'} P(x'\mid x,y=1) V(x')
            \end{aligned}
            \]
            
            \IF{$V^1(x) \le V^0(x)$}
                \STATE $n(s,w,r) \gets \delta$ \COMMENT{Threshold found}
                
                \FOR{$\delta'=\delta$ \TO $\Delta$}
                    \STATE $x'\gets (s, w, \delta', r)$
                    \STATE $V_{\text{new}}(x') = \delta' + \lambda + \sum\limits_{\tilde x} P(\tilde x\!\mid\! x',y\!=\!1) V(\tilde x)$ 
                \ENDFOR
                
                \STATE \textbf{break} \COMMENT{Exit loop over $\delta$}
            \ELSE
                \STATE $V_{\text{new}}(x) \gets V^0(x)$
            \ENDIF
        \ENDFOR
    \ENDFOR

    \STATE $V \gets V_{\text{new}} - V_{\text{new}}(x_{ref})$

\UNTIL{$V$ converges}

\RETURN $n(s, w, r)$
\end{algorithmic}
\end{algorithm}

\begin{algorithm}
\caption{Randomized Policy Bisection Search}
\label{alg:bisection}
\begin{algorithmic}[1]

\STATE Initialize $\lambda^{-} \gets 0$, $\lambda^{+} \gets 1$
\STATE Initialize $\psi_{\lambda^{+}}$ via Alg.~\ref{alg:RVI}
\STATE Initialize $\bar{R}_{\psi_{\lambda^+}}$ applying Theorem~\ref{thm:param} and using~\eqref{eq:barR}

\WHILE{$\bar{R}_{\psi_{\lambda^+}} > R$}
    \STATE $\lambda^{-} \gets \lambda^{+}$
    \STATE $\lambda^{+} \gets 2\lambda^{+}$
    \STATE Compute $\psi_{\lambda^{+}}$ via Alg.~\ref{alg:RVI}
    \STATE Apply Theorem~\ref{thm:param} and compute $\bar{R}_{\psi_{\lambda^{+}}}$~\eqref{eq:barR}
\ENDWHILE

\REPEAT
    \STATE $\lambda^* \gets \frac{\lambda^{+} + \lambda^{-}}{2}$
    \STATE Compute $\psi_{\lambda^*}$ via Alg.~\ref{alg:RVI}
    \STATE Apply Theorem~\ref{thm:param} and compute $\bar{R}_{\psi_{\lambda^*}}$~\eqref{eq:barR}
    \IF{$\bar{R}_{\lambda^*} \ge R$}
        \STATE $\lambda^{-} \gets \lambda^*$
    \ELSE
        \STATE $\lambda^{+} \gets \lambda^*$
    \ENDIF
\UNTIL{$\lambda^*$ converges}
\vspace{0.15\baselineskip}
\STATE $\rho = \dfrac{R - \bar{R}_{\psi_{\lambda^+}}}{\bar{R}_{\psi_{\lambda^-}} - \bar{R}_{\psi_{\lambda^+}}}$
\vspace{0.3\baselineskip}
\STATE Construct $\psi^*$ via Theorem~\ref{thm:cmdp_policy}
\RETURN $\psi^*$
\end{algorithmic}
\end{algorithm}

\begin{algorithm}
\caption{Randomized Single-threshold Policy Bisection Search}
\label{alg:bisection_n}
\begin{algorithmic}[1]

\STATE Initialize $n^{-} \gets 1$, $n^{+} \gets 2$
\STATE Initialize $\bar{R}_{n^+}$ applying Theorem~\ref{thm:param} and using~\eqref{eq:barR}

\WHILE{$\bar{R}_{n^+} > R$}
    \STATE $n^{-} \gets n^{+}$
    \STATE $n^{+} \gets 2n^{+}$
    \STATE Apply Theorem~\ref{thm:param} and compute $\bar{R}_{n^{+}}$~\eqref{eq:barR}
\ENDWHILE

\WHILE{$n^{+}-n^{-}>1$}
    \STATE $n^* \gets \frac{n^{+} + n^{-}}{2}$
    \STATE Apply Theorem~\ref{thm:param} and compute $\bar{R}_{n^*}$~\eqref{eq:barR}
    \IF{$\bar{R}_{n^*} \ge R$}
        \STATE $n^{-} \gets n^*$
    \ELSE
        \STATE $n^{+} \gets n^*$
    \ENDIF
\ENDWHILE
\vspace{0.15\baselineskip}
\STATE $\rho = \dfrac{R - \bar{R}_{n^+}}{\bar{R}_{n^-} - \bar{R}_{n^+}}$
\vspace{0.3\baselineskip}
\STATE Construct $\psi^*$ via Theorem~\ref{thm:cmdp_policy}
\RETURN $\psi^*$
\end{algorithmic}
\end{algorithm}

\end{document}